\documentclass[runningheads]{llncs}
\usepackage[T1]{fontenc}

\usepackage[hidelinks]{hyperref}
\usepackage[utf8]{inputenc}
\usepackage[small]{caption}

\usepackage{graphicx} 
\usepackage{url}
\usepackage{todonotes}
\usepackage{tikz}
\usepackage{subcaption}
\usetikzlibrary{patterns}
\usepackage{amsmath,amssymb}
\usepackage{xspace}
\usepackage{float}
\usepackage[inline]{enumitem}

\usepackage{pifont}

\usepackage{algorithm}
\usepackage{algorithmic}

\newcommand{\langProp}[1]{\ensuremath{\mathcal{L}(#1)}\xspace}	
\newcommand{\true}{\ensuremath{1}}
\newcommand{\false}{\ensuremath{0}}

\newcommand{\modelSet}[1]{\ensuremath{\mathsf{Mod}(#1)}}

\newcommand{\inc}{\ensuremath{\mathcal{I}}}

\newcommand{\algfour}[2]{\ensuremath{\mathsf{Alg4}(#1, #2)}}

\newcommand{\algoneempty}{\ensuremath{\mathsf{Alg1}}}
\newcommand{\algtwoempty}{\ensuremath{\mathsf{Alg2}}}
\newcommand{\algthreeempty}{\ensuremath{\mathsf{Alg3}}}
\newcommand{\algthreepempty}{\ensuremath{\mathsf{Alg3'}}}
\newcommand{\algfourempty}{\ensuremath{\mathsf{Alg4}}}

\newcommand{\both}{\textsf{both}\xspace}

\newcommand{\kb}{\ensuremath{{K}}}

\newcommand{\atoms}{\ensuremath{\mathsf{At}}}

\newcommand{\incdrastic}{\ensuremath{\mathcal{I}_{d}}}

\newcommand{\incc}{\ensuremath{\mathcal{I}_{c}}}

\title{Privacy-Preserving Inconsistency Measurement}

\author{Carl Corea\inst{1} \and
Timotheus Kampik\inst{2} \and
Nico Potyka\inst{3}}
\authorrunning{Corea et al.}
\institute{University of Koblenz, Germany \and
Umeå University, Sweden; SAP, Germany
 \and Cardiff University, UK \\
\email{ccorea@uni-koblenz.de, tkampik@cs.umu.se, PotykaN@cardiff.ac.uk}
}

\begin{document}

\maketitle

\vspace{-.5cm}
\begin{abstract}
    We investigate a new form of (privacy-preserving) inconsistency measurement for multi-party communication. Intuitively, for two knowledge bases $K_A,K_B$ (of two agents $A,B$), our results allow to quantitatively assess the degree of inconsistency for $K_A \cup K_B$ without having to reveal the actual contents of the knowledge bases. Using secure multi-party computation (SMPC) and cryptographic protocols, we develop two concrete methods for this use-case and show that they satisfy important properties of SMPC protocols -- notably, \emph{input privacy}, i.e., jointly computing the inconsistency degree without revealing the inputs.
\end{abstract}

\section{Introduction}
\label{sec:intro}
In multi-agent systems, agents may have to cooperate without being allowed to share their internal knowledge with each other.
For example, revealing internal knowledge or beliefs of an agent may violate (external) privacy requirements, or agents are both cooperating and competing and do not want to reveal knowledge that may give others a competitive advantage. 

In our work, 
we consider multi-agent systems where the agents carry internal knowledge or beliefs in form of propositional logic knowledge bases (KBs), and assume agents may not be allowed to reveal the contents (i.e., formulas) of their KBs to each other. Still, in order to assess the ability to cooperate, it may be necessary for the agents to verify whether, or to what extent, the knowledge bases are \emph{consistent} with each other. Consider the following simplified example from the financial domain:

\begin{example}
Consider two agents $A$ and $B$, with (own) propositional logic KBs $K_A$ and $K_B$ (with agents having knowledge on credit applications, and customers can have different statuses, creditworthiness, or be on a ban list), with:
\begin{align*}
    &K_A=\{\neg(\mathit{banList}\wedge\mathit{creditWorthy})\}\\
    &K_B=\{\mathit{platinumStatus};\,\mathit{platinumStatus}\rightarrow\mathit{creditWorthy};
    \,\mathit{banList}\}
\end{align*}

\end{example}

Clearly, $K_A \cup K_B$ is inconsistent (we will define inconsistency later), and knowing this may be crucial for the agents. 
But as stated, we assume $A$ and $B$ do not want to reveal the formulas in their KBs. 
To solve this issue, we present a novel approach for, what we call, \emph{privacy-preserving inconsistency measurement}. The core idea is that we build on cryptographic protocols from the field of secure multi-party computation, which allow multiple agents to jointly compute a function $f(K_A,K_B)$ without revealing $K_A$ and $K_B$. Here, we propose algorithms to compute (as $f$) different \emph{inconsistency measures} \cite{Thimm:2019a} for $K_A\cup K_B$. 

Our results allow the agents---in a privacy-preserving way---to know
\begin{enumerate*}[label=\roman*)]
    \item whether their knowledge is \emph{consistent}, and
    \item to what \emph{degree} their knowledge disagrees (e.g., wrt. an inconsistency degree, the KBs may still be \emph{sufficiently consistent} s.t. the alignment may be ``good enough'' for collaborating).
\end{enumerate*}
Here, our contributions are as follows:
\begin{itemize}
    \item We present a novel approach for privacy-preserving inconsistency measurement; specifically, for two KBs $K_A,K_B$, we show how to compute two specific inconsistency measures for $K_A\cup K_B$ without revealing $K_A,K_B$ (Section \ref{sec:approaches}). To this aim, we show how private set intersections of sets of KB interpretations can be computed to measure various aspects of (in)consistency.
    \item We evaluate the developed methods by showing important privacy- and runtime complexity properties (Section \ref{sec:discussion}). 
\end{itemize}
We discuss preliminaries in Section~\ref{sec:prelim} and conclude in Section \ref{sec:conlusion}. Proofs are shown in the appendix.

\vspace{-.2cm}
\section{Preliminaries}
\label{sec:prelim}

\subsection{Knowledge Bases, Inconsistency Measurement}
In this work, we consider agents carrying internal knowledge in form of propositional logic knowledge bases. 
For this, let $\atoms$ be some fixed propositional signature, i.\,e., a (possibly infinite) set of propositions, and let $\langProp{\atoms}$ be the corresponding propositional language constructed using the connectives $\wedge$, $\vee$
and $\neg$. 
\vspace{-.4cm}
\begin{definition}\label{def:kb}
    A knowledge base $\kb$ is a finite set of formulas $\kb\subset\langProp{\atoms}$. 
\end{definition}
For a set of formulas $X$, we denote the set of contained propositions as $\atoms(X)$.
An interpretation $\omega$ on \atoms\ is a function $\omega:\atoms\rightarrow\{\false,\true\}$ (where $0$ stands for false and $1$ stands for true). Let $\Omega(\atoms)$ denote the set of all interpretations for \atoms. An interpretation $\omega$ \emph{satisfies} (or is a \emph{model} of) an atom $a\in\atoms$, denoted by $\omega\models a$, iff $\omega(a)=\true$. The satisfaction relation $\models$ is extended to formulas in the usual way. For $\Phi\subseteq\langProp{\atoms}$ we also define $\omega\models \Phi$ if and only if $\omega\models\phi$ for every $\phi\in\Phi$. For a set of formulas $X$, the set of models is 
$\modelSet{X}=\{\omega\in \Omega(\atoms)\mid \omega\models X\}$. 
If $\modelSet{X}=\emptyset$ we write $X\models \perp$ and say that $X$ is \emph{inconsistent}. 

%
An inconsistency \emph{measure} $\inc$ is a function that assigns a non-negative numerical value to a knowledge base.
The concrete behaviour of inconsistency measures is driven by rationality postulates. In this work, we assume inconsistency measures $\inc$ satisfy the basic property of \emph{consistency} (for a KB $K$):
\begin{description}
    \item[Consistency \textsf{CO}] $\inc(\kb)=0$ iff $K\not\models\perp$ 
\end{description}

Numerous inconsistency measures have been proposed (see \cite{Thimm:2019a} for a survey). In this work, we consider the \emph{drastic} inconsistency measure $\incdrastic$ \cite{hunter2008measuring} and the \emph{contension} inconsistency measure $\incc$ \cite{grant2011measuring}, which we define below. In order to define the contension measure we need some additional background on three-valued logic \cite{Priest:1979}. A three-valued interpretation is a function $\nu: \atoms \rightarrow \{0,1,\both\}$, which assigns to every atom either $0$, $1$ or \both, where $0$ and $1$ correspond to $false$ and $true$, respectively, and \both denotes a conflict. 
Assuming the \emph{truth order} $\prec_T$ with $0\prec_T \both \prec_T 1$, the function $\nu$ can be extended to arbitrary formulas as follows: $\nu(\alpha\wedge\beta) = \min_{\prec_T}(\nu(\alpha),\nu(\beta))$, $\nu(\alpha\vee\beta) = \max_{\prec_T}(\nu(\alpha),\nu(\beta))$, $\nu(\neg \alpha)=1$ if $\nu(\alpha)=0$, $\nu(\neg \alpha)=0$ if $\nu(\alpha)=1$, and $\nu(\neg \alpha)=\both$ if $\nu(\alpha)=\both$.
We say an interpretation $\nu$ satisfies a formula $\alpha$, denoted by $\nu \models^3 \alpha$, iff $\nu(\alpha)=1$ or  $\nu(\alpha)= \both$.
We are now ready to define the considered inconsistency measures.

\begin{definition}[Considered Inconsistency Measures]
    Given a knowledge base $K$, define $\incdrastic, \incc$ via:
    \begin{align*}
        \incdrastic(\kb) & = \left\{\begin{array}{cc}                                
				1                & \text{if~} \kb\models\perp                              \\
				0                & \text{otherwise}
				\end{array}\right.&
		\incc(\kb)       & = \min\{|\nu^{-1}(\both)|\mid\nu\models^{3} \kb\}    \\[1ex]
    \end{align*}    
\end{definition}
\vspace{-.8cm}
\begin{example}
    Consider $\kb_1=\{a; a\rightarrow b; \neg b\wedge\neg a; c\}$, then we have that
        $\incdrastic(\kb_1)=1$ and $\incc(\kb_1)=2$.
    (for $\incc$, note that the three-valued interpretation $\nu_1$ with $\nu_1(a)=\both$; $\nu_1(b)=\both$; $\nu_1(c)=1$ is the only three-valued model that assigns $\both$ to a minimal number of atoms). 
\end{example}

In this work---for two knowledge bases $K_A, K_B$ of parties $A,B$, respectively---we compute $\incdrastic(K_A\cup K_B)$ and an upper-bound for $\incc(K_A\cup K_B)$, without $A$ and $B$ having to reveal the contents of their knowledge bases to each other. To allude to some of our results, this will be achieved by comparing the \emph{models} for the individual knowledge bases (also in a privacy-preserving way). For example, we can exploit that $\modelSet{\kb_1}\cap\modelSet{\kb_2}=\emptyset$ iff $\kb_1\cup\kb_2 \models \perp$, which allows to verify consistency without revealing formulas. Intuitively, the interpretations should also not be revealed, which we will show how to handle. In the following subsections, we discuss important notions and methods from a security perspective.

\subsection{Cryptographic Techniques}
In this work, we consider (asymmetric) encryption schemes, or cryptosystems, that can securely encode and decode messages with algorithmic techniques.

\begin{definition}[Cryptosystem, \cite{sen2013homomorphic}]
    Let $\mathcal{M}$ be a set of messages, called a \emph{message space}, and let $\rho \in \mathbb{N}$ be a \emph{security parameter}. Then, an encryption scheme is a tuple $(\textbf{K},\textbf{E},\textbf{D})$, where
    \begin{itemize}
        \item $\textbf{K}$ is a (key generation) function that takes the security parameter $\rho$ and returns a key pair $(k_e,k_d)$ for encryption/decryption, with $k_e\in\mathcal{K}_e, k_d\in\mathcal{K}_d$ (with $\mathcal{K}_e,\mathcal{K}_d$ being key spaces).
        \item $\textbf{E}$ is an (encryption) function $\textbf{E}:\mathcal{K}_e\times \mathcal{M}\rightarrow \mathcal{C}$ that returns a ciphertext for a plaintext $m$, where $\mathcal{C}$ is a ciphertext space.
        \item $\textbf{D}$ is a (decryption) function $\textbf{D}:\mathcal{K}_d\times\mathcal{C}\rightarrow\mathcal{M}$ that takes a ciphertext and outputs a plaintext $m$, s.t. if $c=\textbf{E}(k_e,m)$ then \emph{Probability}$[\textbf{D}(k_d,c)\neq m]$ is negligible, i.e., \emph{Probability}$[\textbf{D}(k_d,c)\neq m]\leq 2^{-\rho}$.
    \end{itemize}
\end{definition}

We consider encryption functions that are probabilistic, i.e., $\textbf{E}$ can return different ciphertexts even for two equal inputs (on the other hand, \textbf{D} is deterministic) \cite{sen2013homomorphic}. To clarify, given two plaintexts $m_1=m_2=1$ and a key pair $k_e,k_d$ produced by the encryption scheme, we have that \emph{Probability}$[\textbf{E}(k_e, m_1)=\textbf{E}(k_e, m_2)]\leq 2^{-\rho}$, but $\textbf{D}(k_d,\textbf{E}(k_e, m_1))=\textbf{D}(k_d,\textbf{E}(k_e, m_2))$. 
This is also referred to as the encryption scheme being  IND-CPA secure (ciphertext indistinguishability under chosen plaintext attacks) \cite{cramer2015secure}.
We use $\textbf{E}(k_e, m_1)\equiv\textbf{E}(k_e, m_2)$ to denote that two ciphertexts carry ``semantically" the same value, even though the ciphertexts are not identical. 

For the technical development of our techniques, we assume two communicating parties, where both parties act via a \emph{honest-but-curious} adversarial model \cite{cramer2015secure}, i.e., parties do not deviate from the protocol but may try to infer additional information from the data they obtain. We comment on the effects of a party taking on other adversarial models in Section \ref{sec:discussion}. 

\subsection{Secure Multi-Party Computation}
An important cryptographic field we build on is that of secure multi-party computation (SMPC) \cite{cramer2015secure}. The goal of SMPC approaches is to allow multiple parties $P_1,...,P_n$ to compute a function $f$ over their (respective) inputs $x_1,...,x_n$, without revealing the inputs to the other parties. Any protocol performing this computation should satisfy the following properties:
\begin{description}
    \item[Input Privacy (\textsf{IP})] Inputs should not be revealed during computation.
    \item[Correctness (\textsf{Cor})] The revealed output is the actual result of $f(x_1,...,x_n)$.
\end{description}
We use the term ``privacy-preserving" to denote that a computation satisfies \textsf{IP}. 


To develop privacy-preserving inconsistency measurement techniques, we will devise protocols that can compute (various aspects of) intersections of sets of knowledge base models while satisfying \textsf{IP, Cor}. This type of protocol is referred to as private set intersection (PSI) (cf. Section \ref{subsec:PSI}). PSI protocols build on so-called homomorphic encryption schemes, which we introduce next.

\subsection{Homomorphic Encryption}
\label{subsec:homomorph-enc}

Homomorphic encryption \cite{yi2014homomorphic} is a cryptographic method that allows certain mathematical operations to be performed directly on encrypted data. Specifically, we consider \emph{fully homomorphic} encryption schemes \cite{gentry2009fully}, which allow addition and multiplications of numerical values while remaining encrypted.


\begin{definition}[Homomorphic encryption scheme, \cite{sen2013homomorphic}] 
    Let $\mathcal{M}$ be a message space, and let $\rho$ be a security parameter. Then, a homomorphic encryption scheme is a quadruple $(\textbf{K},\textbf{E},\textbf{D},\circ)$ as follows:
    \begin{itemize}
        \item $\textbf{K}, \textbf{E}, \textbf{D}$ are key-generation-, encryption- and decryption functions as before.
        %
        %
        \item $\circ$ is an operator for which it holds that for all messages $m_1,m_2\in\mathcal{M}$: if $m_3=m_1\circ m_2$, and $c_1=\textbf{E}(k_e,m_1)$, and $c_2=\textbf{E}(k_e,m_2)$, then \emph{Probability}$[\textbf{D}(k_d,c_1 \circ c_2)\neq m_3]$ is negligible.
    \end{itemize}    
\end{definition}
In other words, a homomorphic encryption scheme is an encryption scheme with the property that the operation $\circ$ is correctly preserved when performing it on the encrypted ciphertexts themselves.

In the following, we assume encryption schemes that are fully homomorphic (i.e., where $\circ$ can be either addition (+) or multiplication ($\times$)), s.t. for all $m_1,m_2\in\mathcal{M}, k_e\in\mathcal{K}_e$ we have (cf. \cite{yi2014homomorphic}):
    \begin{align*}
        &\textbf{E}(k_e,m_1) + \textbf{E}(k_e,m_2) \equiv \textbf{E}(k_e,m_1+m_2),
        &m_1 + \textbf{E}(k_e,m_2) \equiv \textbf{E}(k_e,m_1+m_2),\\
        &\textbf{E}(k_e,m_1) \times \textbf{E}(k_e,m_2) \equiv \textbf{E}(k_e,m_1\times m_2),
        &m_1 \times \textbf{E}(k_e,m_2) \equiv \textbf{E}(k_e,m_1\times m_2).
    \end{align*}

\begin{example}
    Let a key pair $k_e,k_d$ produced by a fully homomorphic encryption scheme. For $m=5$ and $c=\textbf{E}(k_e,m)+2$, we have $\textbf{D}(k_d,c)=7$.
\end{example}

\begin{remark}\label{remark:poly}
    We make the standard assumption that the size of the ciphertexts remains polynomially bounded (cf. the property of circuit-privacy in \cite{armknecht2015guide}). This 
    ensures that the ciphertext obtained by performing an operation $\circ$ on two inputs is hard to distinguish from a ciphertext obtained by encrypting a plaintext $m$.
\end{remark}

This extends to vectors via element-wise operations.
A survey of fully homomorphic encryption schemes can be found in \cite{yi2014homomorphic}. 

\subsection{Private Set Intersection}\label{subsec:PSI}
PSI protocols \cite{morales2023private} are subtypes of SMPC that allow to compute  the intersection of two sets without revealing the rest of the sets. There are many applications for PSI, for example, finding (only) the common friends in the contact lists of two parties without disclosing the full contact lists to each other. As an example, consider the following baseline PSI protocol.

\begin{example}
Let two parties $A,B$ each with a set containing exactly one integer, respectively, $x$ and $y$. To compute PSI, $A$ employs a fully homomorphic encryption scheme to generate a key pair $k_e,k_d$ and sends $\textbf{E}(k_e,x)$ to $B$. Then, $B$ computes $c = r*(\textbf{E}(k_e,x) - y)$ (which is a ciphertext, cf. Remark \ref{remark:poly}), where $r$ is a random nonzero integer chosen by $B$. $A$ now computes $\textbf{D}(k_d,c)$: If the result is 0, $x$ and $y$ are identical, otherwise, $\{x\}\cup \{y\}=\emptyset$. In the latter case, $A$ cannot infer any information about $y$ (as $r$ is chosen by $B$). In any case, $B$ cannot infer anything about $x$ (beyond the guarantees of the encryption scheme) as $B$ only obtains ciphertext. The protocol can be performed symmetrically.
\end{example}

In this work, we consider versions of PSI protocols where only the size of the intersection is revealed. As we will show, inconsistency can then be characterized with various aspects of intersection sizes for sets of interpretations. 

\section{Approaches for Privacy-Preserving Inconsistency Measurement}\label{sec:approaches}

For the remainder, we fix two parties $A,B$ with respective knowledge bases $K_A,K_B$. Also, we fix a fully homomorphic encryption scheme $(\textbf{K},\textbf{E},\textbf{D},\circ)$ and a corresponding key-pair $k_e,k_d$. The goal of this section is then two-fold: first, we develop an SMPC protocol allowing to compute $\incdrastic(K_A\cup K_B)$; then, to allow for a more gradual measure, we develop an SMPC protocol allowing to compute an upper-bound for $\incc(K_A\cup K_B)$ (both protocol satisfying \textsf{IP} and \textsf{Cor}). For the remainder, we assume the KBs $K_A,K_B$ on their own are consistent (and the task is to assess the consistency of $K_A\cup K_B$). 

From \textsf{IP}, it is immediate that no formulas must be revealed for the computations. 
Instead, in our protocols, the parties will exchange their respective models; as discussed, consistency can then be verified by checking whether $\modelSet{K_A}\cap\modelSet{K_B}\neq \emptyset$. An important remark here is that the interpretations/models should also not be revealed in plain form. Otherwise, it would be possible to disjunctively write each interpretation as a conjunction of atoms, which would yield a formula in disjunctive normal form that is equivalent to the KB. Thus, we introduce PSI-based protocols that allow to \emph{privately} compute the intersection of $\modelSet{K_A}$ and $\modelSet{K_B}$. For this, we need some further notation.

For any interpretation $\omega$ over $\atoms$, we will encode $\omega$ as a bit sequence of $1s$ and $0s$, which indicates the
truth value of the atoms 
in alphabetical order.
\begin{example}
    Let $\atoms=\{a,b,c\}$; then we write $\omega_1 = 101$ to encode  
        $\omega_1(a)=1$, and  $\omega_1(b)=0$, and $\omega_1(c)=1$.
\end{example}
This encoding will be used in various encryption processes. For example, we can directly encrypt interpretations by encrypting the encoding -- in the example, $\textbf{E}(k_e, 101)$.
%
A further remark is that, also wrt. Axioms (1)-(4), this shorthand notation is useful for comparing interpretations via a bitwise comparison. For example, for two interpretations 11 and 10, a bitwise comparison 01 indicates that the first digit is identical and the second digit differs.

We are now ready to define a core protocol for comparing \emph{two} interpretations. This core protocol will then later be used in two subsequent protocols (for $\incdrastic, \incc$).

\subsection{General Protocol for Privacy-Preserving Comparison of (Two) Interpretations}


Assume two parties $A$ and $B$ who each have one interpretation ($\omega_A, \omega_B$).
For example, this interpretation could be derived from their knowledge bases under the closed-world assumption (if an atom is not entailed by the knowledge base, it is supposed to be false). Importantly, the two parties agree on a shared set of atoms $\atoms$ (needed to produce the encoding of the interpretations).
We now want to verify if these interpretations are compatible. More precisely, Algorithm~\ref{alg:computeNumDifferentAssignmentsBetweenTwoInterpretations} specifies an SMPC protocol that takes as input two interpretations $\omega_A, \omega_B$ and returns the number of atoms that they interpret differently.
For the protocol, we recall the introduced encoding that allows to represent interpretations as binary numbers. For any binary number $w$, let $\mathit{len}(w)$ denote the number of digits of $w$, and $w_i$ the $i^{th}$ bit of $w$.

\begin{algorithm}[H]
\caption{($\algoneempty$) Compute Number of Differing Truth Assignments Given Interpretations $A, B$}
\label{alg:computeNumDifferentAssignmentsBetweenTwoInterpretations}
\begin{algorithmic}[1]
\REQUIRE Shared set of atoms $\atoms$, Interpretations $\omega_A, \omega_B$ (over $\atoms$, in shorthand notation)
\ENSURE Number of truth assignments differing between $\omega_A$ and $\omega_B$

\STATE $A$ generates key pair $(k_e, k_d)$
\STATE $A$ generates a vector $v=\langle \omega_{A_1},...,\omega_{A_{\mathit{len}(\omega_A)}} \rangle$
\STATE $A$ computes $v_{enc}=\langle \textbf{E}(k_e,\omega_{A_1}),...,\textbf{E}(k_e\omega_{A_{\mathit{len}(\omega_A)}}) \rangle$
\STATE $A$ sends $v_{enc}$ to $B$
\STATE $B$ generates a vector $v_B=\langle \omega_{B_1},...,\omega_{B_{\mathit{len}(\omega_B)}} \rangle$
\STATE $B$ computes $v_{A\oplus B} = v_{enc}-v_{B}$ \hfill \COMMENT{``XOR" operation; B cannot read result}
\STATE $B$ computes $n = \sum_{i=1}^{|v_{A\oplus B}|} (v_{A\oplus B}[i])^2$ \COMMENT{B cannot read result}
\STATE $B$ sends $n$ to $A$
\STATE $A$ computes $n' = \textbf{D}(k_d, n)$
\RETURN $n'$
\end{algorithmic}
\end{algorithm}

\begin{example}
    Assume two parties $A,B$ with interpretations $\omega_A=110$ and $\omega_B=101$, respectively.
    \begin{itemize}
        \item A generates a key pair.
        \item A generates the vector $v=\langle 1,1,0\rangle$, resp. $v_{enc}=\langle \textbf{E}(k_e,1), \textbf{E}(k_e,1), \textbf{E}(k_e,0)\rangle$.
        \item B generates the vector $v_B=\langle1,0,1\rangle$.
        \item B computes $v_{A\oplus B}=v_{enc}-v_B =\\\langle \textbf{E}(k_e,1)-1,\textbf{E}(k_e,1)-0,\textbf{E}(k_e,0)-1\rangle$ (note that each position is still a ciphertext that B cannot read).
        \item B computes $n = \sum_{i=1}^{|v_{A\oplus B}|} (v_{A\oplus B}[i])^2$ (sum of absolute values; result is a ciphertext).
        \item A receives $n$ and returns $n' = \textbf{D}(k_d,n) = 2$        
    \end{itemize}
    In result, $A$ can infer that $\omega_A,\omega_B$ differed in 2 assignments, but, importantly, cannot infer at which assignments the interpretations differ. The protocol can be performed symmetrically to produce the result for $B$.
\end{example}

\begin{theorem}
    Algorithm \ref{alg:computeNumDifferentAssignmentsBetweenTwoInterpretations} satisfies $\textsf{IP}, \textsf{Cor}$.
\end{theorem}


The presented core protocol allows to correctly compare the number of differing truth assignments for two interpretations, without revealing them. 
This core protocol will now be leveraged for computing further measures, specifically, by extending the protocol to take as input not only one interpretation, but two respective sets of interpretations/models by $A$ and $B$.

\subsection{Privacy-Preserving Computation of $\incdrastic$}
We recall the definition of $\incdrastic$ and parties $A,B$ with knowledge bases $K_A,K_B$. To compute $\incdrastic$, we build on the fact that $\modelSet{K_A}\cap\modelSet{K_B}=\emptyset$ iff $K_A\cup K_B\models\perp$. To leverage Algorithm \ref{alg:computeNumDifferentAssignmentsBetweenTwoInterpretations} for computation, we create two bit sequences as follows:
First, $A,B$ agree on a shared set of atoms $\atoms$. Then, both parties independently create a truth table showing satisfaction of their knowledge base, where the truth-assignments (rows) are in ascending binary order, ensuring equally ordered and exhaustive enumeration of all possible truth value combinations.
\begin{example}\label{ex:truthTables}
    Consider two agents $A$ and $B$, each with their KBs $K_A$ and $K_B$, respectively, with
        $K_A=\{a \land b\}$ and $K_B=\{\neg a\}$.
    Then, $A$ and $B$ construct the following truth tables:
    \[
    A:
    \begin{array}{c|c|c}
    a & b & \models K_A \\
    \hline
    0 & 0 & 0 \\
    0 & 1 & 0 \\
    1 & 0 & 0 \\
    1 & 1 & 1 \\
    \end{array}
    \quad
    B:
    \begin{array}{c|c|c}
    a & b & \models K_B \\
    \hline
    0 & 0 & 1 \\
    0 & 1 & 1 \\
    1 & 0 & 0 \\
    1 & 1 & 0 \\
    \end{array}
    \]
\end{example}
Both parties then take the last column, which is a bit-sequence. In the example, we get the two sequences $S_A=0001$ and $S_B=1100$. It is important to note that (as the row index for both tables correspond due to the ascending order), each index $i$ over both sequences exactly encodes whether the $i^{th}$ assignment of truth values satisfies $K_A$, resp., $K_B$. We then leverage these sequences as input for Algorithm \ref{alg:computeNumDifferentAssignmentsBetweenTwoInterpretations} to verify consistency. For this, we define a slight variation of Algorithm \ref{alg:computeNumDifferentAssignmentsBetweenTwoInterpretations}, denoted Algorithm \ref{alg:computeNumDifferentAssignmentsBetweenTwoInterpretations}$^{binary}$, where we change Lines 6 and 7 as follows:
\begin{enumerate*}[label=\roman*)]
    \item 6: $B \text{ computes } v_{A\oplus B} = 1- (v_{enc}*v_{B})$;
    \item 7: $B \text{ computes } n = \prod_{i=1}^{|v_{A\oplus B}|} (v_{A\oplus B}[i])^2$.
\end{enumerate*}

This has the following impact on the algorithm output: For line 6, if the two multiplicants differ or are both 0 (corresponding to inconsistency or non-satisfaction), the entry computed in line 6 is 1 (indicating a distance). Otherwise, it is 0 (indicating correspondence and satisfaction). In result, if at least one row satisfies both knowledge bases, the result of Algorithm \ref{alg:computeNumDifferentAssignmentsBetweenTwoInterpretations}$^{binary}$ is simply 0 (cf. line 7). If all rows differ, the returned value in this way is exactly 1, hiding the number of rows. This binary version of Algorithm \ref{alg:computeNumDifferentAssignmentsBetweenTwoInterpretations} can be leveraged as follows.

\begin{algorithm}[H]
\caption{($\algtwoempty$) Compute $\incdrastic(K_A\cup K_B)$}
\label{alg:computeID}
\begin{algorithmic}[1]
\REQUIRE Shared set of atoms $\atoms$, knowledge bases $K_A,K_B$
\ENSURE $\incdrastic(K_A\cup K_B)$

\STATE $A,B$ create respective truth tables in ascending binary order (over $\atoms$).
\STATE $A,B$ obtain (via the last column) their private sequences $S_A, S_B$.

\STATE $A,B$ compute $d=$ Algorithm \ref{alg:computeNumDifferentAssignmentsBetweenTwoInterpretations}$^{binary}$ wrt. $\atoms, S_A, S_B$
\RETURN $d$
\end{algorithmic}
\end{algorithm}
\vspace{-.4cm}
\begin{example}
    We recall the KBs and truth tables from Example \ref{ex:truthTables}, with the two sequences $S_A=0001, S_B=1100$ (this relates to lines 1-2 of \algtwoempty). Then, Algorithm \ref{alg:computeNumDifferentAssignmentsBetweenTwoInterpretations}$^{binary}$ is computed wrt. $S_A, S_B$ and stored as $d$. Regarding that Algorithm, recall that if (the assignments) of at least row match and satisfy both knowledge bases, $d=0$, and $d=1$ otherwise. In the example, indices 1,2,4 do not match, and, while the third index (i.e., the interpretation $\omega(a)=1;\omega(b)=0$) matches, this interpretation is not a model. In turn, the output is 1.
\end{example}


\begin{theorem}
    Algorithm \ref{alg:computeID} satisfies \textsf{IP}, \textsf{Cor}.
\end{theorem}




\subsection{$\text{Privacy-Preserving Approximation of }\incc$}

%
We continue with a gradual measure based on $\incc$.
For this, recall parties $A,B$ with $K_A, K_B$. Then, consider $\algthreeempty$, which works over \emph{models} of $K_A$ and $K_B$:
\begin{algorithm}[H]
\caption{($\algthreeempty$) Compute Smallest Distinct Number of Mismatching Assignments for Any Pair in Two Sets of Models}
\begin{algorithmic}[1]
\REQUIRE Shared set of atoms $\atoms$, two sets of models: $\modelSet{\kb_A}$, $\modelSet{\kb_B}$
\ENSURE Smallest number of different assignments over all combinations of models

\STATE $A$ initializes an empty set $S = \{\}$
\FORALL{elementA $\in \modelSet{K_A}$}
    \FORALL{elementB $\in \modelSet{K_B})$}
        \STATE Both $A$ and $B$ perform \textbf{Algorithm 1} wrt. $\atoms$, elementA, elementB.
        \STATE $A$ stores the result in $S$
    \ENDFOR
\ENDFOR
\RETURN MIN($S$)
\end{algorithmic}
\end{algorithm}

\begin{example}\label{ex:incC}
    Assume two parties A,B (with $K_A, K_B$) and $\modelSet{K_A}=\{111,110\}$, $\modelSet{K_B}=\{100,101\}$.
    \begin{itemize}
        \item Alg. 3 will iterate over all combinations of models and yields:
             $\mathit{Alg1}(111,100)= 2$;
             $\mathit{Alg1}(111,101)= 1$;             
             $\mathit{Alg1}(110,100)= 1$;
             $\mathit{Alg1}(110,101)= 2$.
         \item The returned result is MIN$(\{1,2\})=1$.
    \end{itemize}
\end{example}

First, observe that the result of Example \ref{ex:incC} can be interpreted s.t. no combination of models agree ($\modelSet{\kb_A\cup\kb_B}=\emptyset$). 
We now show how the result of Algorithm 3 can be used to approximate the contension measure. The idea is that we use the results of Algorithm 3 to derive a set of \emph{three-valued} interpretations $\mathit{Mod}_3(\kb_A\cup\kb_B)$ (to plug into $\incc$).
For this, the following will be useful.
\begin{lemma}
\label{lemma_adding_B}
Let $F$ be a formula and let $I$ be a three-valued interpretation that satisfies $F$.
Let $J$ be a three-valued interpretation obtained from $I$ by changing the
interpretation of a single atom to $\both$. Then
$I(F) = J(F)$ or $J(F)=\both$.
\end{lemma}

\begin{corollary}\label{cor:two-three-valued-models}
Let $I$ be a two-valued model of $K$ and let $J$ be a three-valued interpretation obtained from $I$ by changing the
interpretation of atoms to $\both$. Then $J$ is a three-valued model of $K$.
\end{corollary}

For two KBs $\kb_1, \kb_2$, we now use this result to define a function that can transform the two-valued interpretations $Mod(K_1), Mod(K_2)$ into a set of three-valued interpretations $M_3$ s.t. for every $m\in M_3: m\models_3 \kb_1\cup\kb_2$.

\begin{definition}
    Let two knowledge bases $\kb_1, \kb_2$ and let two interpretations $i_1, i_2$ over $\atoms$ s.t. $i_1\models\kb_1$, $i_2\models\kb_2$. Then, define a three-valued interpretation via the function $f(i_1,i_2)$, where
    \[
        f(i_1, i_2)(j) =
        \begin{cases}
        i_1(j), & \text{if } i_1(j) = i_2(j) \\
        \both, & \text{otherwise}
        \end{cases}
    \]
\end{definition}
Corollary \ref{cor:two-three-valued-models} implies that a three valued interpretation obtained via $f$ is a three-valued model for $\kb_1\cup\kb_2$.
With a slight abuse of notation, we 
let:
\begin{align*}
   &{} f(Mod(\kb_1), Mod(\kb_2)) = \quad \{ f(i_1, i_2) \mid i_1 \in Mod(\kb_1), i_2 \in Mod(\kb_2) \}. 
\end{align*}

We now show the relationship of Algorithm 3 to $\incc$.
\begin{proposition}
Let $x$ be the inconcistency value computed by Algorithm 3. Then $x \geq \incc(\kb_A \cup \kb_B)$.
\end{proposition}
However, 
Algorithm 3 can overestimate the inconsistency value.
\begin{example}
Let $\kb_A = \{a, a \rightarrow b_1 \wedge b_2\}$ and $\kb_B = \{a, a \rightarrow \neg b_1 \wedge \neg b_2\}$.
Then $\modelSet{\kb_A} = \{111\}$ and  $\modelSet{\kb_B} = \{100\}$. Hence, $f(\modelSet{\kb_A},\modelSet{\kb_B}) = \{1\both\both\}$ and 
Algorithm 3 will return $2$. However, $\both 00$, $\both 01$, $\both 10$, $\both 11$ are also three-valued models 
of $\kb_A \cup \kb_B$. Therefore $\incc$ is $1$.
\end{example}
While Algorithm 3 cannot compute the contension inconsistency value exactly, it will never underestimate the inconsistency 
(it gives an upper bound on the inconsistency value). Importantly, it will also never report a positive inconsistency value
when the knowledge bases are, in fact, mutually consistent.
\begin{proposition}
If $\kb_A \cup \kb_B$ is consistent, then Algorithm 3 will return $0$.
\end{proposition}

Let us note that Alg. 3 satisfies our privacy guarantee. 
\begin{theorem}
    Algorithm 3 satisfies $\textsf{IP}$. 
        
\end{theorem}


While Algorithm 3 satisfies \textsf{IP}, it
reveals more than the output\footnote{This is also referred to as the SMPC properties of confidentiality (\textsf{Con}).}: Clearly, a) $B$ learns the number of models provided by $A$, and b) $A$ learns all distinct differences in truth assignments over all model combinations.
We therefore show a slight variation ($\algfourempty$) which allows to counteract this. 

\begin{algorithm}
\caption{($\algfourempty$) Compute Smallest Distinct Number of Mismatching Assignments for Any Pair in Two Sets of Models, Satisfying \textsf{Con}}
\begin{algorithmic}[1]
\REQUIRE Shared set of atoms $\atoms$, two sets of models: $\modelSet{\kb_A}$, $\modelSet{\kb_B}$
\ENSURE Smallest number of different assignments over all combinations of models
\STATE $A$ initializes $\modelSet{K_A}$' as multiset of length $|2^\atoms|$, containing all and only elements of $\modelSet{K_A}$
\STATE $B$ initializes an empty list $S = <>$
\FORALL{elementA $\in \modelSet{K_A}'$}
    \FORALL{elementB $\in \modelSet{K_B})$}
        \STATE Both $A$ and $B$ perform \textbf{Algorithm 1} wrt. $\atoms$, elementA, elementB.
        \STATE $B$ stores the (encrypted) result in $S$
    \ENDFOR
\ENDFOR
\STATE $B$ initializes an empty list $L = <>$
\FORALL{$i \in \{0, ..., |\atoms|\}$}
    \STATE $B$ computes $L_i \leftarrow \prod_{d \in S} (i - d)$
\ENDFOR
\FORALL{$i \in \{0, ..., |\atoms|\}$}
    \STATE $B$ computes $p \leftarrow$ random prime number
    \STATE $B$ computes $L_i \leftarrow (\prod_{0 ... i} L_i)^{p - 1}$
\ENDFOR
\STATE $A$ decrypts $L$
\RETURN Index of the first element in $L$ that is $0$
\end{algorithmic}
\end{algorithm}

In line 1, $A$ creates a multiset with models of $K_A$ of size $2^{|\atoms|}$ (possibly containing duplicates). This is a padding that ensures $A$ does not reveal the number of models. Lines 2-8 are analogous to $\algthreeempty$, using the padded multiset. Lines 9-12 are an encrypted computation by $B$. $L$ is a list from 0 to $|\atoms|$, the possible range for $\incc$. Then, we compute the distance between every $L_i$ and all results in $S$: $i - d = 0$ must hold for at least one $d \in S$. We are only interested in the minimum of $0, ..., |\atoms|$ that is a match (lines 13-18), hence we obscure all left of the minimum by means of prime encryption and set all right of the minimum to zero (numbers that are not zero are ``meta-encrypted'' and stay encrypted even after $A$ decrypts the result). The smallest distinct number of mismatching assignments is the first index in $L$ of an element that is $0$. This solves the problems with $\textsf{Con}$ exhibited in $\algthreeempty$: i) $B$ cannot know the number of $A$'s models as $A$ sends a padded multiset, ii) $A$ does not learn all distinct differences in truth assignments over all combinations.

\begin{example}
    Recall Example \ref{ex:incC} with  $\modelSet{K_A}=\{111,110\}$, $\modelSet{K_B}=\{100, \\ 101\}$. First, $A$ will create a padded list of size $2^{|\atoms|}$, here: $\modelSet{K_A}'=\{111,110, \\ 111, 110,111,110,111,110\}$. $A$ and $B$ perform lines 3-8 which yields a list of (encrypted) differences: $\algoneempty$ is performed for all (model, model)-tuples in $\modelSet{K_A}' \times \modelSet{_B}$, yielding the (encrypted) list $\{\textbf{E}(k_e,2), \textbf{E}(k_e,1),...\}$. $B$ then checks which of the potential distances in $\{0, 1, 2, 3\}$ exist in this list by executing lines 9-12, yielding the list $\langle 2, 0, 0, 2 \rangle$ (again, encrypted). The multiplication of every element in the list with its predecessors and subsequent prime encryption (by $B$) in lines 13-14 results in the list $\langle \mathsf{enc}_p, 0, 0, 0 \rangle$, where $\mathsf{enc}_p$ is the prime-encrypted $2$. Finally, $A$ decrypts the list. However, $\mathsf{enc}_p$ is useless (as it was meta-encrypted). $A$ can only infer from $\langle \mathsf{enc}_p, 0, 0, 0 \rangle$ that $\incc = 1$ (index of the first $0$).
\end{example}

%
%
    
While $\algfourempty$ comes with improvement wrt. $\textsf{Con}$, $A$ has to create a padded multiset of exponential size. In the following we discuss this trade-off.

\section{Discussion}\label{sec:discussion}

We start by showing the upper and lower bounds of runtime-complexity in Table \ref{tab:runtimecomplexity} (proofs in appendix).
\begin{table}
\centering
\begin{tabular}{|c|c|c|}
\hline
 & $\mathcal{O}$ & $\Omega$\\ \hline
Algorithm 1 & $\mathcal{O}(k+|\atoms|)$ & $\Omega(k+|\atoms|)$   \\ 
Algorithm 2 & $\mathcal{O}(k+2^{|\atoms|})$  &$\Omega(k+2^{|\atoms|})$   \\ 
Algorithm 3 & $\mathcal{O}(k+2^{2^{|\atoms|}}*|\atoms|)$  &  $\Omega(k+|\atoms|)$\\ 
Algorithm 4 &  $\mathcal{O}(k+2^{2^{|\atoms|}}*|\atoms|)$  & $\Omega(k+{2^{|\atoms|}}*|\atoms|)$ \\ \hline
\end{tabular}
\caption{Runtime-complexity (upper ($\mathcal{O}$)/lower ($\Omega$)) of the developed algorithms wrt. $\atoms$; $k$ = cost of key generation.}
\label{tab:runtimecomplexity}
\end{table}

While $\algthreeempty$ violates $\textsf{Con}$, it can approach polynomial scaling in best-cases. 
$\algfourempty$ retains its exponential component due to padding, trading off complexity and privacy requirements, depending on the use-case. $\algtwoempty$ scales exponentially.

For the discussion of algorithms we have considered an \emph{honest-but-curious} adversarial model. Threats in this setting are bounded by the compliance with $\textsf{IP}$. 
For adversarial models such as \emph{malicious adversary} (participants may deviate from the protocol), the guarantees given by $\textsf{IP}$ also hold. However, intuitively, such a threat model can affect the correctness of the results: A malicious adversary can deliberately provide altered models of his/her own knowledge base, e.g., flipping all bits (note that the ability of the adversary to manipulate ciphertexts is mitigated by IND-CPA security). Likewise, the adversary could provide models even if the KB is inconsistent. While preventing the adversary to provide fake models cannot be mitigated, various methods exist to prove the consistency of the own KB (without revealing it) via zero-knowledge proofs \cite{bellare2012foundations}, which can be put before our algorithms if needed. 

One should also be aware of risks by repeated queries. Given no restrictions wrt. the number of queries sent, $A$ could 
reveal information about $B$'s KB by altering the input for different queries. For this risk, it is important to consider what we actually reveal.
Accordingly, we observe the worst-case probabilities with which one agent can successfully guess a model in another agents' KB.
\begin{proposition}
    $A$ can correctly guess a model in $K_B$ with a probability of at least $\frac{1}{|\modelSet{K_A}|}$ if $K_A$ and $K_B$ are consistent; if $K_A$ and $K_B$ are inconsistent, $A$ can correctly guess with a probability of at least $\frac{1}{|\Omega(\atoms) \setminus \mathsf{Mod}_{<\algfourempty}|}$, where $\mathsf{Mod}_{<\algfourempty} := \{m | m \in \Omega(\atoms),  \algfour{\modelSet{K_a}}{\{m\}} < \algfour{\modelSet{K_A}}{\modelSet{K_B}}\}$.
\end{proposition}
For both cases, $A$ may straightforwardly reveal a formula equivalent to $B'$s KB by measuring $g(\{m\}, \modelSet{K_B})$ for all $m \in \Omega(\atoms))$, where $g \in \{\algtwoempty, \algthreeempty, \algfourempty\}$. 

\section{Conclusion}\label{sec:conlusion}
We have introduced novel methods for privacy-preserving inconsistency measurement. By leveraging SMPC and homomorphic encryption, the proposed algorithms enable agents to collaboratively evaluate the consistency of their KBs without revealing sensitive information. 
While the approach successfully implements input privacy, it also highlights trade-offs in runtime complexity and potential risks in adversarial settings. Overall, the framework advances the state of the art by enabling cooperative inconsistency measurement in privacy-critical settings. 
Future work can, for example, focus on methods for privately computing interpolants/common knowledge or disagreement between $>2$ parties \cite{potyka2018measuring,ribeiro2020measuring}.

\bibliographystyle{splncs04}
\bibliography{refs}

\section*{Appendix: Proofs for Technical Results}

\setcounter{theorem}{0}
\setcounter{proposition}{0}
\setcounter{lemma}{0}
\setcounter{corollary}{0}

\begin{theorem}
    Algorithm \ref{alg:computeNumDifferentAssignmentsBetweenTwoInterpretations} satisfies $\textsf{IP}, \textsf{Cor}$.
    \begin{proof}
        The protocol adheres to \textsf{IP} as no party learns anything beyond the Hamming distance, where the Hamming distance itself is the agreed-upon output (in particular, B works on ciphertexts only in line 6 and 7). The potential to infer complete input information in edge cases (e.g., output = 0 or $|\atoms|$) is intrinsic to the meaning of the output and not an additional leakage under typical definitions. It is simply a characteristic of the output itself.     
        For \textsf{Cor}, proceed by invariance: at each step, the algorithm accurately tracks the number of differing truth assignments between \(\omega_A\) and \(\omega_B\). Initially, the vectors \(v\) and \(v_B\) represent the truth assignments of \(\omega_A\) and \(\omega_B\), respectively. The XOR operation in \(v_{A \oplus B}\) ensures that each entry reflects whether the corresponding truth assignments differ (\(1\)) or are identical (\(0\)), even under encryption. The summation step correctly accumulates the total number of differing assignments - squaring each \(v_{A \oplus B}[i]\) does not change its value since \(v_{A \oplus B}[i] \in \{0, 1\}\). Finally, decryption reveals the correct count without altering the result. This invariant holds throughout the algorithm, ensuring its correctness. 
    \end{proof}
\end{theorem}

\begin{theorem}
    Algorithm \ref{alg:computeID} satisfies \textsf{IP}, \textsf{Cor}.
    \begin{proof}
        The protocol adheres to \textsf{IP} as no party learns anything beyond a binary assessment of consistency, where the assessment itself is the agreed-upon output. \textsf{Cor} is established by maintaining an invariant:   
        At each step, the algorithm correctly tracks whether there exists at least one model that satsifies both knowledge bases.
        In line one, both parties create sequences encoding satisfaction, where each index is in the same order and corresponds to the same truth value assignment(s). In line 6, for every such index, Algorithm 1$^{\mathit{binary}}$ returns 0, if the values are both 1, and 1, otherwise. Then via multiplication, if there exists at least one 0, the product is also 0. This ensures that the algorithm will return 0 if there is at least one interpretation satisfying both $K_A, K_B$, and 1, otherwise.
    \end{proof}
\end{theorem}

\begin{lemma}
Let $F$ be a formula and let $I$ be a three-valued interpretation that satisfies $F$.
Let $J$ be a three-valued interpretation obtained from $I$ by changing the
interpretation of a single atom to $\both$. Then
$I(F) = J(F)$ or $J(F)=\both$.
\end{lemma}
\begin{proof}
We prove the claim by structural induction. We can assume w.l.o.g. that we change an atom that is contained in the formula $F$
because changing the truth value of another atom cannot affect the interpretation by truth-functionality of the logical 
connectives.

For the base case, assume that $F$ is an atom.
If we change the interpretation of $F$ to $\both$, we have $J(F) = \both$.

Since all formulas can be expressed using only $\neg$ and $\wedge$, it is sufficient to consider these cases for the induction step.

Consider $F = \neg G$.
If $I(\neg G) = 0$, then $I(G) = 1$ and by the induction assumption, we have $J(G) = 1$ or $J(G) = \both$. Thus,
$J(\neg G) = 0 = I(\neg G)$ or $J(\neg G) = \both$.
If $I(\neg G) = 1$, then $I(G) = 0$ and by the induction assumption, we have $J(G) = 0$ or $J(G) = \both$. Thus,
$J(\neg G) = 1 = I(\neg G)$ or $J(\neg G) = \both$.
If $I(\neg G) = \both$, then $I(G) = \both$ and by the induction assumption, we have $J(G) = \both = I(G)$.

Consider $F = G_1 \wedge G_2$. 
If $I( G_1 \wedge G_2) = 0$, then
$I( G_1 ) = 0$ or $I( G_2 ) = 0$.
Assume w.l.o.g. that $I(G_1) = 0$ (the case $I(G_2) = 0$ is analogous). 
By the induction assumption, $J( G_1 ) = 0$ or $J( G_1 ) = \both$.
Hence, $J( G_1 \wedge G_2) = 0$ or $J( G_1 \wedge G_2) = \both$.
If $I( G_1 \wedge G_2) = 1$, then
$I( G_1 ) = 1$ and $I( G_2 ) = 1$.
By the induction assumption, ($J( G_1 ) = 1$ or $J( G_1 ) = \both$) and ($J( G_2 ) = 1$ or $J( G_2 ) = \both$).
Hence, $J( G_1 \wedge G_2) = 1$ or $J( G_1 \wedge G_2) = \both$.
If $I( G_1 \wedge G_2) = \both$, then
$I( G_1 ) = \both$ or $I( G_2 ) = \both$.
Assume w.l.o.g. that $I(G_1) = \both$ (the case $I(G_2) = \both$ is analogous). 
By the induction assumption, $J( G_1 ) = \both$.
Hence, $J( G_1 \wedge G_2) = \both$.
\end{proof}

\begin{corollary}\label{cor:two-three-valued-models}
Let $I$ be a two-valued model of $K$ and let $J$ be a three-valued interpretation obtained from $I$ by changing the
interpretation of atoms to $\both$. Then $J$ is a three-valued model of $K$.
\end{corollary}
\begin{proof}
For all formulas $F \in K$, we have $I(F) = 1$ by assumption.
Note that two-valued interpretations are a special case of three-valued
interpretations. 
When $J$ changes the truth value of $k$ atoms to $\both$, it can be seen
as a sequence $J_0, J_1, \dots, J_k$ of three-valued models, 
$J_0 = I, J_k = J$ and $J_{i}$ is obtained from $J_{i-1}$, $i=1,\dots,k$ by changing the interpretation of a single atom to $\both$.
Hence, Lemma \ref{lemma_adding_B} guarantees that $J_i(F) = 1$ or $J_i(F) = \both$ for all $i=1,\dots,k$. Hence, $J$ satisfies $K$.
\end{proof}

\begin{proposition}
Let $x$ be the inconcistency value computed by Algorithm 3. Then $x \geq \incc(\kb_A \cup \kb_B)$.
\end{proposition}
\begin{proof}
Corollary \ref{cor:two-three-valued-models} guarantees that the 3-valued interpretations computed by Algorithm 3 are models of
three-valued models of $\kb_A \cup \kb_B$. Since $x$ is the minimal number of $\both$ found across these models, and 
$ \incc(\kb_A \cup \kb_B)$ is the minimal number across all models, we must have $x \geq \incc(\kb_A \cup \kb_B)$.
\end{proof}

\begin{proposition}
If $\kb_A \cup \kb_B$ is consistent, then Algorithm 3 will return $0$.
\end{proposition}
\begin{proof}
If $\kb_A \cup \kb_B$ is consistent, there must be an $i \in \modelSet{\kb_A \cup \kb_B}$. Hence, Algorithm 3
will compute $f(i,i) = i$ and therefore return $0$.
\end{proof}

\begin{theorem}
    Algorithm 3 satisfies $\textsf{IP}$. 
    \begin{proof}
        The protocol adheres to \textsf{IP} as no party learns the concrete models. $A$ only learns the number of differing truth assignments but cannot map this to specific models. 
        
    \end{proof}
\end{theorem}

\begin{proposition}
    The runtime complexity of $\algoneempty$ is $\mathcal{O}(k+|\atoms|)$, where $k$ is the asymptotic cost of generating the key pair.
    \begin{proof}
        We proceed by line. Line 1 has cost $k$ (see above). Line 2 and 3 perform constant operations over all $|\atoms|$ positions of the interpretation ($\mathcal{O}(|\atoms|)$). Line 4 is only part of the communication. Line 5 is analogous to line 2. Subtracting two vectors of length $|\atoms|$ (line 6) is $\mathcal{O}(|\atoms|)$. Note that $v_{A\oplus B}$ is also of size $|\atoms|$. Line 7 performs a constant operation over $|\atoms|$ positions. Line 8 is analogous to 4. Line 9 is constant. Thus we have $\mathcal{O}(k+5*|\atoms|+1)=\mathcal{O}(k+|\atoms|)$.
    \end{proof}    
\end{proposition}

\begin{proposition}
    The lower-bound runtime complexity of $\algoneempty$ is $\Omega(k+|\atoms|)$, where $k$ is the asymptotic cost of generating the key pair.
    \begin{proof}
        Analogous to $\mathcal{O}$.
    \end{proof}    
\end{proposition}

\begin{proposition}
    The runtime complexity of $\algtwoempty$ is $\mathcal{O}(k+2^{|\atoms|})$, where $k$ is the asymptotic cost of generating the key pair.
    \begin{proof}
        Straightforward from $\algoneempty$ (Note that both parties construct a bit-sequence of length $n$, where $n$ is $2^{|\atoms|}$).
    \end{proof}
\end{proposition}

\begin{proposition}
    The lower-bound runtime complexity of $\algtwoempty$ is $\Omega(k+2^{|\atoms|})$, where $k$ is the asymptotic cost of generating the key pair.
    \begin{proof}
        Analogous to $\mathcal{O}$.
    \end{proof}    
\end{proposition}

\begin{proposition}
    Let $A$ have $s$ models and $B$ have $t$ models. The runtime complexity of $\algthreeempty$ is $\mathcal{O}(k+s*t*|\atoms|)$ (assuming the key is only generated once), where $k$ is the asymptotic cost of generating the key pair. As both parties could have up to $2^{|\atoms|}$ models, this relates to $\mathcal{O}(k+2^{2^{|\atoms|}}*|\atoms|)$. 
    \begin{proof}
        The outer loop iterates over all $s$ models in \(\modelSet{\kb_A}\). The inner loop iterates over all $t$ models in \(\modelSet{\kb_B}\). For each pair of models (each of length $|\atoms|$), $\algoneempty$ is called, which has a complexity of \(\mathcal{O}(k + |\atoms|)\). 
    \end{proof}
\end{proposition}

\begin{proposition}
    Let $A$ have $s$ models and $B$ have $t$ models. 
    The lower-bound runtime complexity of $\algthreeempty$ is $\Omega(k+|\atoms|)$, where $k$ is the asymptotic cost of generating the key pair.
    \begin{proof}
        In general the costs are $(k+s*t*|\atoms|)$ (assuming the key is only generated once), where $k$ is the cost of generating the key pair. In the best case, both parties have 1 model each (recall the KBs are consistent per assumption). This relates to $\Omega(k+1*1*|\atoms|)$. 
    \end{proof}    
\end{proposition}

\begin{proposition}
    Let $A$ have $s$ models and $B$ have $t$ models. The runtime complexity of $\algfourempty$ is $\mathcal{O}(k+2^{2^{|\atoms|}}*|\atoms|)$ (assuming the key is only generated once), where $k$ is the asymptotic cost of generating the key pair. 
    \begin{proof}
        Lines 1-8 are analogous to $\algthreeempty$ ($\mathcal{O}(k+2^{2^{|\atoms|}}*|\atoms|)$). Note both parties could have up to $2^{|\atoms|}$ models, so $S$ can have a size of $2^{|\atoms|}$.        
        Line 9 is constant. Line 10 calls $|\atoms|$ times an operation that requires $2^{|\atoms|}$ subtractions. Lines 13-16 run $|\atoms|$ times. Line 17 is constant. So we have $\mathcal{O}(k+2^{2^{|\atoms|}}*|\atoms|+1+|\atoms|*2^|\atoms|+|\atoms|)=\mathcal{O}(k+2^{2^{|\atoms|}}*|\atoms|)$.
    \end{proof}
\end{proposition}

\begin{proposition}
    Let $A$ have $s$ models and $B$ have $t$ models. 
    The lower-bound runtime complexity of $\algfourempty$ is $\Omega(k+{2^{|\atoms|}}*|\atoms|))$, where $k$ is the asymptotic cost of generating the key pair.
    \begin{proof}
        $A$ pads to $2^|\atoms|$ models, so even if $B$ has 1 model only, line 5 is called $2^|\atoms|$ times, where line 5's cost is $|\atoms|$. Lines 9-17 are anaologous to $\mathcal{O}$.
    \end{proof}    
\end{proposition}

\end{document}